\documentclass[]{revtex4-1}

\usepackage[utf8x]{inputenc}
\usepackage[british]{babel}

\usepackage{float}

\usepackage{color}

\usepackage{subcaption}
\usepackage{enumerate}

\usepackage{microtype}

\usepackage{amsfonts,amsmath,amssymb}
\usepackage{amsthm}

\usepackage{hyperref}
\hypersetup{
  pdfauthor={Friedrich Bös; Etienne Vouga; Omer Gottesman; Max Wardetzky},
  pdftitle={On the incompressibility of cylindrical origami patterns},
}

\usepackage{graphicx}
\usepackage{wrapfig}
\graphicspath{{./}}

\makeatletter
\newtheorem*{rep@theorem}{\rep@title}
\newcommand{\newreptheorem}[2]{%
\newenvironment{rep#1}[1]{%
 \def\rep@title{#2 \ref{##1}}%
 \begin{rep@theorem}}%
 {\end{rep@theorem}}}
\makeatother

\newtheorem{corollary}{\bf Corollary}[section]

\theoremstyle{plain}
\newtheorem{thm}{Theorem}
\newreptheorem{thm}{Theorem}
\newtheorem*{thm*}{Theorem}
\newtheorem{lemma}[thm]{Lemma}
\newtheorem*{lemma*}{Lemma}

\newtheorem*{definition*}{Definition}

\newtheorem*{obs*}{Observation}

\theoremstyle{remark}
\newtheorem{rem}{Remark}
\newtheorem*{rem*}{Remark}

\DeclareMathOperator{\II}{II}

\begin{document}

\title{On the incompressibility of cylindrical origami patterns}
\author{Friedrich Bös}
\email[]{f.boes@math.uni-goettingen.de}
\affiliation{Institute for Numerics and Applied Mathematics, University of Göttingen}
\author{Etienne Vouga}
\affiliation{Department of Computer Science, The University of Texas at Austin, Austin, Texas 78712, USA}
\author{Omer Gottesman}
\affiliation{School of Engineering and Applied Sciences, Harvard University, Cambridge, Massachusetts 02138, USA}
\author{Max Wardetzky}
\affiliation{Institute for Numerics and Applied Mathematics, University of Göttingen}





\begin{abstract}
  The art and science of folding intricate three-dimensional structures out of paper has occupied artists, designers, engineers, and mathematicians for decades, culminating in the design of deployable structures and mechanical metamaterials. Here we investigate the axial compressibility of origami cylinders, i.e., cylindrical structures folded from rectangular sheets of paper. We prove, using geometric arguments, that a general fold pattern only allows for a finite number of \emph{isometric} cylindrical embeddings.
  Therefore, compressibility of such structures requires either stretching the material or deforming the folds. Our result
  considerably restricts the space of constructions that must be searched when designing new types of origami-based rigid-foldable deployable structures and metamaterials.
\end{abstract}

\maketitle
\section{Introduction}
\label{sec:intro}
Consider the experiment of crushing a soda can: axially compressing a thin cylinder. It has long been known that (i) such compression results in diamond-shaped (Yoshimura) crease patterns
~\cite{Yoshimura,Coppa1966} and (ii) crushing an ideal cylinder unavoidably induces in-plane stress in the cylinder. However, various questions about the global rigidity and \emph{flexibility} of such thin-walled cylinders and their idealized counterparts, the rigid origami cylinders, still remain poorly understood despite the fact that considerable effort has been put into the study of planar collapsible, or rigid-foldable, origami~\cite{Tachi2009}.
Furthermore, there exist numerous examples of origami cylinders that \emph{appear} to be truly rigidly foldable---deformable
purely isometrically, without any in-plane strain, some examples of which are shown in Figure~\ref{fig:teaser} and the supplemental video.
But is this really so?


\begin{figure}[H]
  \centering
  \includegraphics*[width=\textwidth]{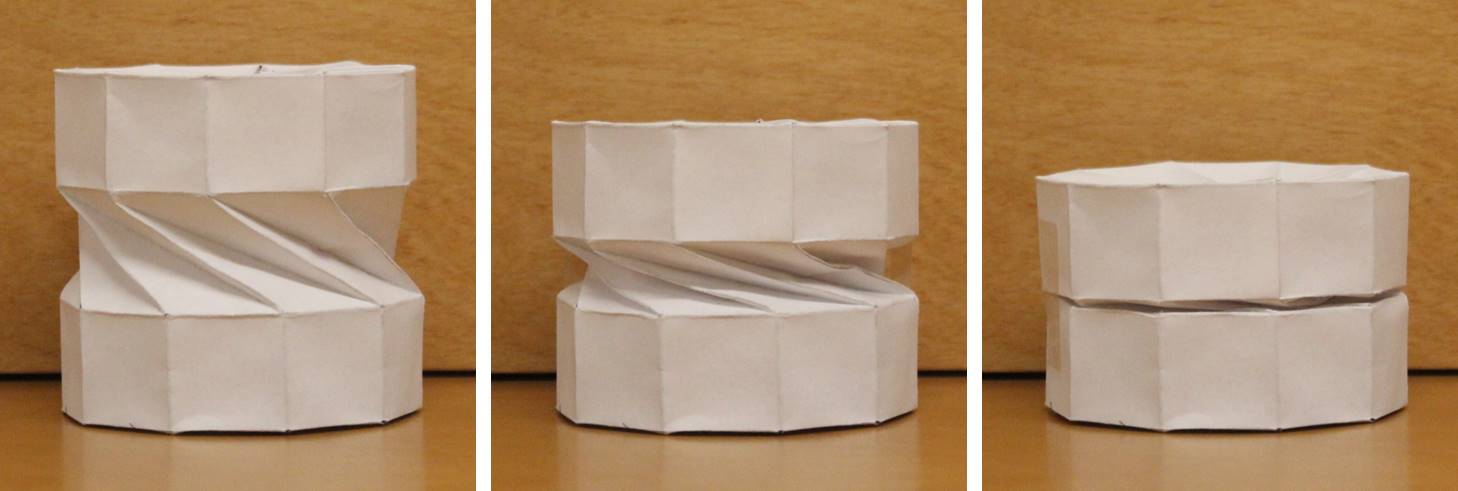}
  \caption{Compressing an origami cylinder made from ordinary paper suggests that it can rigidly collapse. But is this indeed so?}
   \label{fig:teaser}
\end{figure}

The study of rigid foldability for \emph{closed} surfaces, i.e., polyhedral surfaces without holes, has a long tradition in mathematics. For instance, Cauchy's theorem~\cite{cauchyrecherche,Connelly1979} asserts that the boundary surface of a convex polyhedron is always rigid; moreover, the famous \emph{bellows theorem}~\cite{Connelly1998} states that a flexible closed surface must maintain its enclosed volume during isometric deformation. Applied to the problem of foldable origami cylinders, the bellows theorem rules out rigid foldability \emph{provided} that the cylinder's top and bottom do not alter shape when being deformed (since then the cylinder could be sealed with end caps to form a closed surface).

On the other hand, the work of Guest and Pellegrino~\cite{Guest1994a} informs us that purely isometric deformations \emph{are} indeed possible \emph{if} one (theoretically) allowed for sliding of folds within the material, which, however, cannot be considered origami foldability. Only recently, Yasuda and Yang~\cite{Yasuda2015}, building on the work of Tachi~\cite{Tachi2009}, have constructed fold patterns for origami cylinders that allow for true rigid foldability \emph{without} moving folds within the material. In this paper, we show that
the pattern obtained by Yasuda and Yang is very special in its ability to compress origami cylinders isometrically (i.e., without introducing in-plane strain): we show that fold patterns that do not include vertical fold lines are incompressible.

The key to our argument is to formulate the problem of rigid foldability of origami cylinders as a root finding problem for a certain \emph{real analytic} function (to be specified below). We show that this function has at least one nonzero value, implying that it can only have isolated zeroes, which rules out continuous isometric deformations of origami cylinders. Our result extends the bellows theorem by proving that it is impossible to crush a large class of origami cylinders (e.g., those shown in Figure~\ref{fig:nonbellows}) without the restricting assumption that the cylinder's top and bottom maintain their shape during deformation.

\section{When are origami cylinders compressible?}
\label{sec:cans}
We formalise our framework by defining a \emph{fold pattern} in the $(uv)$-plane as a finite set of straight \emph{fold lines} drawn on a rectangular domain of dimensions $l \times h$. See for instance the top row of Figure~\ref{fig:schwarzcirc} for an example.

An \emph{origami cylinder} is made from such a fold pattern by gluing together its left and right boundaries and folding the paper along each fold line. More precisely, origami cylinders are isometric embeddings of fold patterns into 3D $(xyz)$-space, such that (i) fold lines are mapped to straight line segments, (ii) the embedding is piecewise four times continuously differentiable\footnote{For details about this assumption we refer to the appendix.
}, and (iii) all horizontal lines in the material domain map to closed curves lying in planes perpendicular to the $z$ axis. This latter assumption is motivated by the prevalence of rotational symmetry among apparently-compressible origami cylinders, such as the one in Figure~\ref{fig:teaser}. The second assumption implies (as shown in the appendix) that there are no folds besides those that are explicitly defined by the fold lines. 

Notice that these assumptions include fold patterns such as those considered by Yasuda and Yang in~\cite{Yasuda2015}, which are known to be isometrically compressible. The central result of the present work is that isometric compressibility necessarily requires \emph{vertical} fold lines, i.e., fold lines that are perpendicular to the upper and lower boundary of the rectangular domain. Indeed, the absence of vertical fold lines necessarily leads to origami cylinders that are \emph{not} isometrically deformable:

\begin{thm}
  An origami cylinder whose fold pattern does not contain vertical fold lines has at most a finite number of heights at which it can be isometrically embedded. Hence, the absence of vertical folds prohibits isometric deformations with continuously varying heights.

  \label{thm:compressibility}
\end{thm}
An additional consequence of our analysis poses constraints on fold patterns that contain vertical fold lines. We show that the corresponding folds in the embedded cylinder are either flat or must have dihedral angle $\pm\pi$ and thus form sharp creases---a striking feature of the collapsible origami cylinders discovered by Yasuda and Yang~\cite{Yasuda2015}.
\begin{figure}[t]
  \centering
  \newcommand{\factor}{0.493}
  \vspace{0pt}
  \begin{minipage}[t]{\factor\linewidth}
	\centering
	\vspace{0pt}
	\includegraphics[width=\linewidth]{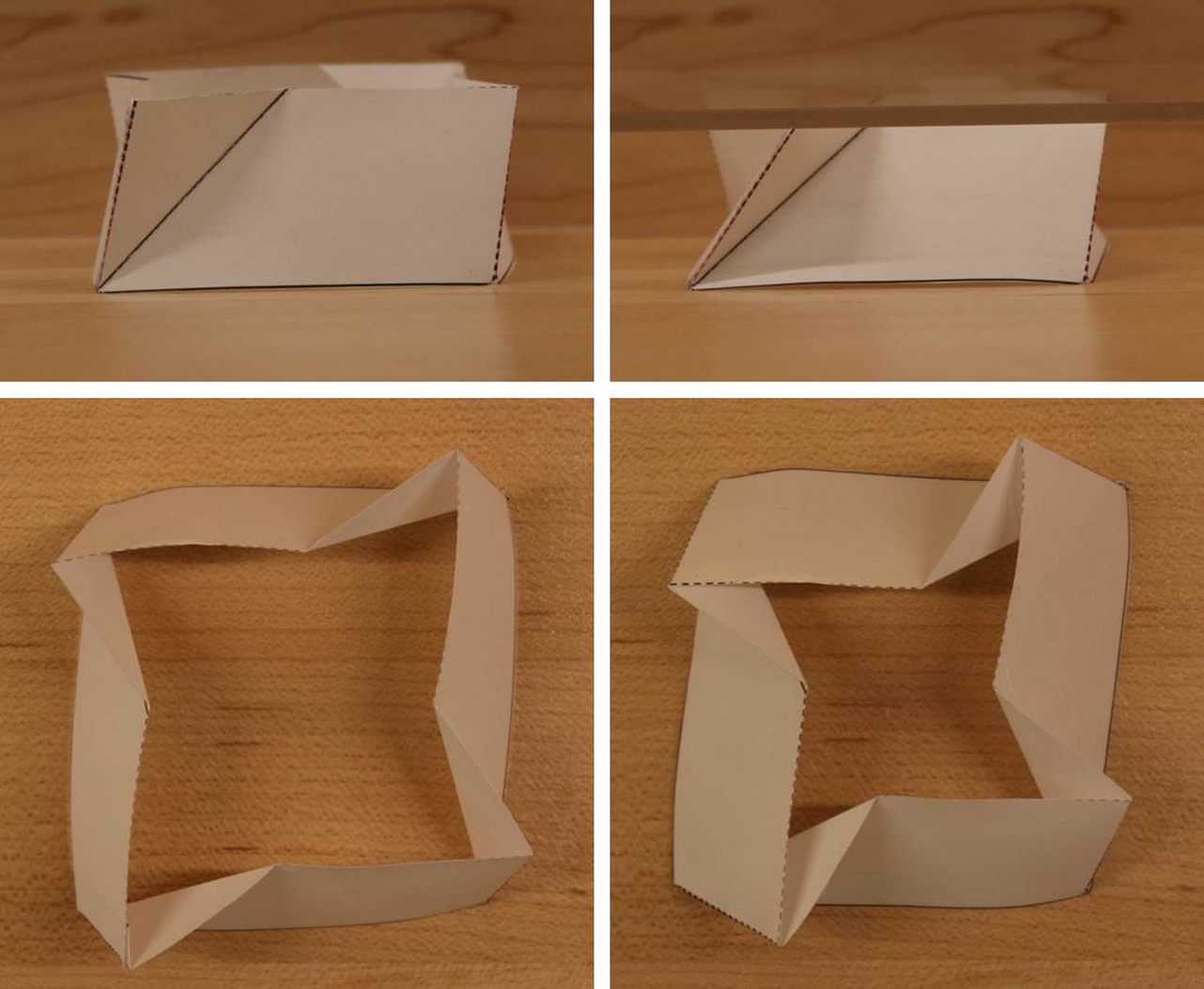}
	\caption{An origami cylinder to which the bellows theorem does not apply. Its apparent collapsibility is disproven by our result.}
	\label{fig:nonbellows}
  \end{minipage}
  \hfill
  \vspace{0pt}
  \begin{minipage}[t]{\factor\linewidth}
	\vspace{0pt}
	\centering
	\includegraphics[width=\linewidth]{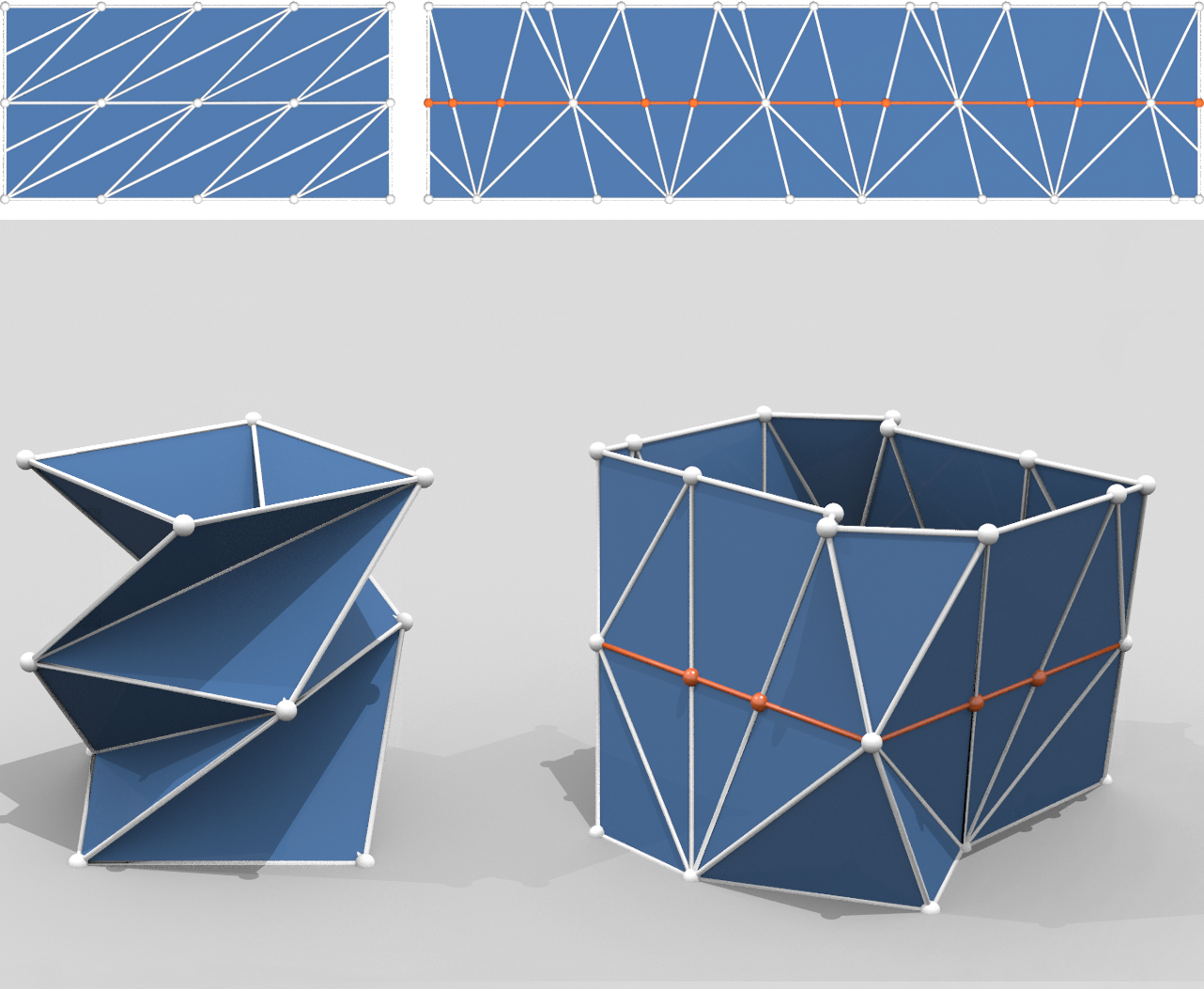}
	\caption{Two origami cylinders with regular (left) and irregular (right) fold patterns (white lines) and illustration of strip construction (orange lines).}
	\label{fig:schwarzcirc}
  \end{minipage}
\end{figure}

The remainder of this work is dedicated to proving these results. Our proof relies on three key insights.
First, a candidate fold pattern can be split into a stack of horizontal
\emph{strips}
that contain no fold lines that intersect each other in the strip's interior (see Figure~\ref{fig:fanexample} and the top row of Figure~\ref{fig:schwarzcirc}). 
\begin{figure}[!t]
	\centering
	\newcommand{\factor}{0.49}
	\includegraphics[width=\linewidth]{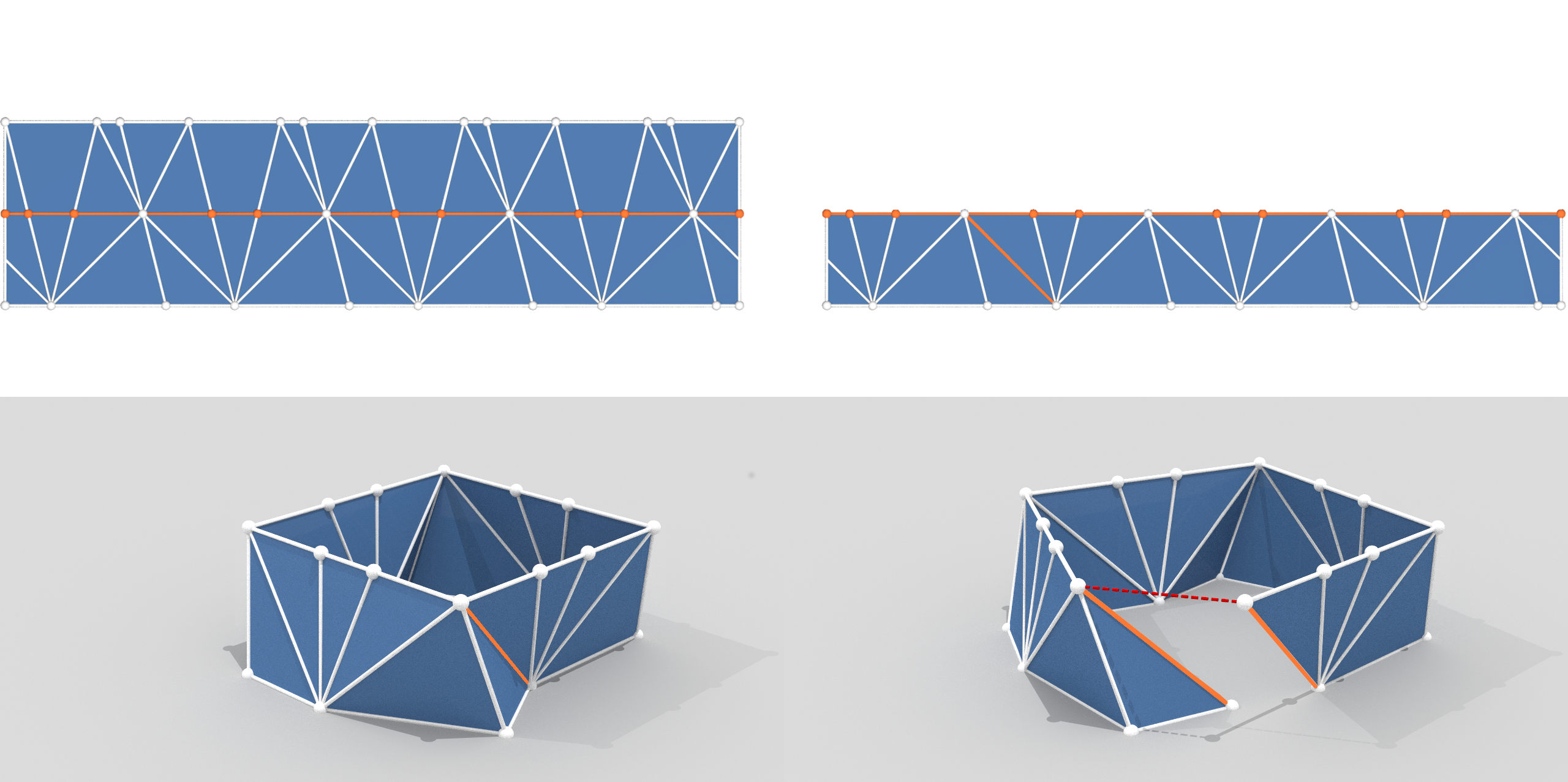}
	\caption{Top row: a fold pattern is split into strips, with the bottom one in detail. Bottom row: The lowermost strip is cut open along the fold marked in orange. For different heights, the embedded strip may be closed (left) or open (right). For open strips, we measure the \emph{gap} (the length of the dashed red line) of either boundary.}
	\label{fig:fanexample}
\end{figure}
Compressibility of the entire fold pattern can then be reduced to compressibility of each of the pattern's strips.

Secondly, we investigate the problem of embedding a single strip as an origami cylinder with prescribed height in 3D space. We show that such an embedding is not always possible for every prescribed height. However, it \emph{is} possible to embed any strip for every prescribed height \emph{if} one drops the assumption that the embedded strip must close up into a cylinder, while maintaining the remainder of assumptions (i)--(iii) above. Furthermore, we prove that the number of embeddings of such open strips with prescribed height is finite. 

Finally, we measure for each embedding of an open strip the failure to close up into a cylinder. This is done by evaluating the distance between corresponding points on both sides of the cut. We define two \emph{gap functions} measuring this distance at the upper and lower boundary of the strip (see Figure~\ref{fig:fanexample}).
Clearly, an embedding closes up into a cylinder if and only if both gap functions vanish. We prove that for any continuous vertical compression of an open strip, the upper and lower gap functions are  \emph{analytic} functions of the embedded height.
As we show, these functions cannot be identically zero in the absence of vertical folds. Far from being isometrically compressible, a given fold pattern can thus be embedded at only a discrete set of heights as an origami cylinder.

\section{Embedding a single strip}
\label{sec:construction}
In order to analyse a strip's gap functions, we first show that every strip's boundary is a planar polygon. We then show how to compute the turning angles of this polygon as an \emph{analytic} function of the prescribed height of the embedded strip. Once these angles are known, the strip's upper and lower gap functions can directly be read off since edge lengths of each boundary segment are known a priori due to the requirement of isometry.

The fact that every strip's upper and lower boundaries are planar polygons is a consequence of a more general result that is of interest on its own. Starting with a fold pattern, from which an origami cylinder is constructed, notice that our assumptions (i)--(iii) above only require straight embedded fold lines but do not explicitly impose conditions on the surface in between. Yet the following holds:
\begin{thm}
  \label{thm:flatness}
  Every origami cylinder whose embedded height is strictly less than the material height consists entirely of planar faces that meet at the embedded fold lines. In particular, there are no creases in the embedding apart from those given by the fold lines.
\end{thm}
Notice that the requirement that the embedded height $H$ is strictly less than the material height $h$ is essential. Indeed, an ordinary (round) cylinder with $H=h$ is an origami cylinder (albeit without folds) in our sense. As a corollary of this theorem, by triangulating each of the planar faces, we can assume without loss of generality that a fold pattern has triangular faces only.

As a second consequence of this theorem, the horizontal cross sections of an origami cylinder are planar polygons whose edge lengths are determined by the requirement of isometry. The remainder of this section is devoted to explicit expressions of the turning angles of these polygons, which we require in order to analyse the gap function mentioned above. 

To this end, we first decompose fold patterns into single \emph{strip fold patterns}. Denote by \emph{vertices} the intersections of fold lines with each other or the domain's boundary.  Strips then arise from cutting along the horizontal lines passing through each of the fold pattern's vertices,
see Figure~\ref{fig:schwarzcirc}.
Notice that if a fold pattern is isometrically compressible in our sense, then at least one of its strips must be compressible as well. Conversely, if no strip is compressible, then the entire fold pattern is incompressible.

We then examine the possible embeddings of a strip with material height $h$ and prescribed embedded height $H<h$. Here we also consider \emph{open} embeddings of the strip, i.e., those that do not close up into a cylinder but where horizontal lines in the fold pattern still map to curves perpendicular to the $z$-axis, see Figure~\ref{fig:fanembedding}.
Theorem~\ref{thm:flatness} can be readily extended to open embeddings of a strip, implying that the embedded strip's upper and lower boundaries are planar polygonal lines. We show that for the case $H<h$ the turning angles of these two polygonal lines are each constrained to a finite set of values that \emph{analytically} depend on $H$ and the fold pattern only.

In order to relate the prescribed embedded height $H$ to the turning angles of the embedded strip's boundaries,  we require the notion of \emph{fans}. A fan is a \emph{maximal} subset of the material that is bounded between fold lines that emanate from a common vertex (the \emph{apex}), see Figure~\ref{fig:fanexample}. We call a fan \emph{downward} if the apex lies on the top of the strip; otherwise, we call it \emph{upward}. Notice that downward and upward fans alternate and that two consecutive fans share a common fold line.
We call the two outermost fold lines and vertices of a fan \emph{exterior}, the others \emph{interior}, see Figure~\ref{fig:fanexample}. 

\subsection{Turning angles at interior vertices}
In what follows we treat angles to lie in the interval $(-\pi, \pi]$. Accordingly, all equalities of this section that are concerned with angles are to be understood modulo $2\pi$. 

We first focus on determining the turning angles at interior vertices of the polygonal boundary of an embedded fan. We only treat downward fans here; upward fans can be dealt with similarly. As before, denote by $H$ the embedded height and by $h$ the material height of the fan.
Let $a$ denote the apex of the fan in the material domain and label the fan's bottom vertices $p_1$ through $p_n$. Denote by $A$ the embedded position of $a$ and let $P_i$ denote the embedded position of $p_i$ for $i=1,\dots, n$, see Figure~\ref{fig:fanembedding}. Suppose for simplicity that $A$ lies on the $z=H$ plane, that $P_1$ through $P_n$ lie in the $(xy)$-plane, and denote by $\bar{A}$ the vertical projection of $A$ onto the $(xy)$-plane. Then we have:
\begin{figure}[t]
  \centering
  \newcommand{\factor}{0.48}
  \begin{minipage}[t]{\factor\linewidth}
	\centering
	\includegraphics[width=\linewidth]{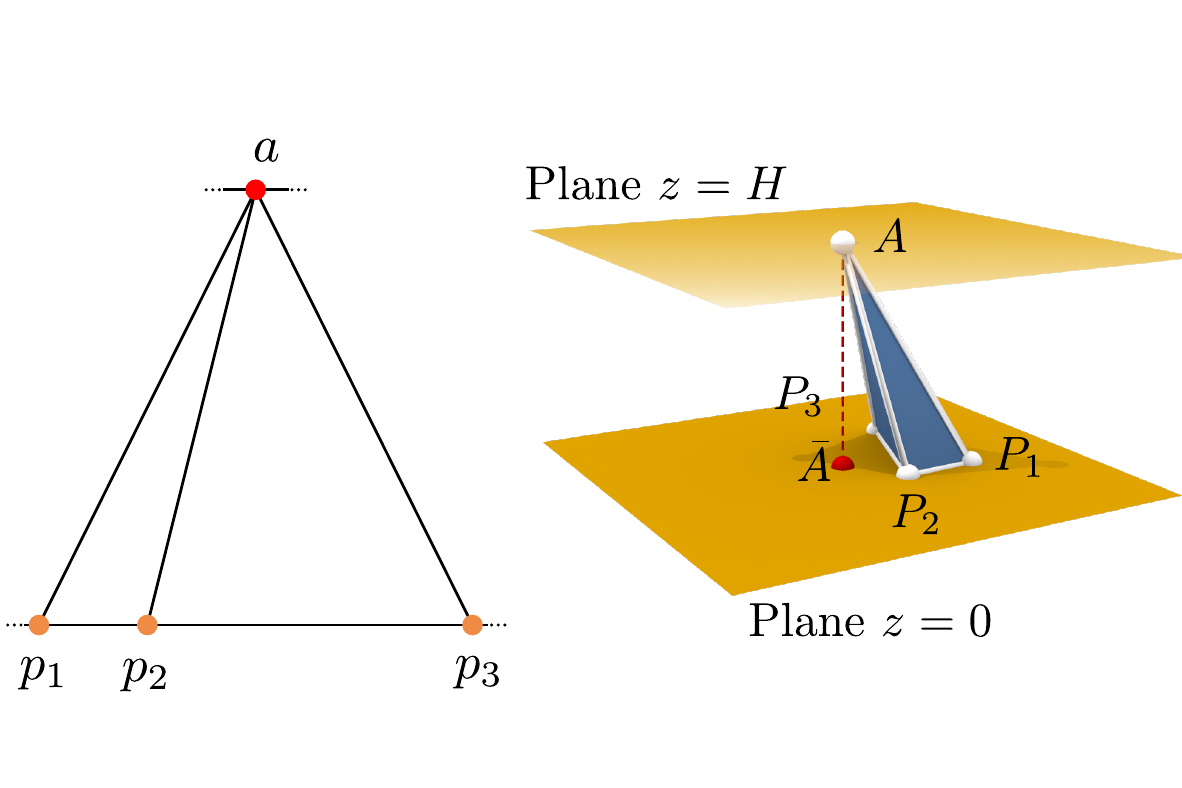}
	\caption{A fan with 4 vertices is embedded (right image)}
	\label{fig:fanembedding}
  \end{minipage}
  \hfill
  \begin{minipage}[t]{\factor\linewidth}
	\centering
	\includegraphics[width=\linewidth]{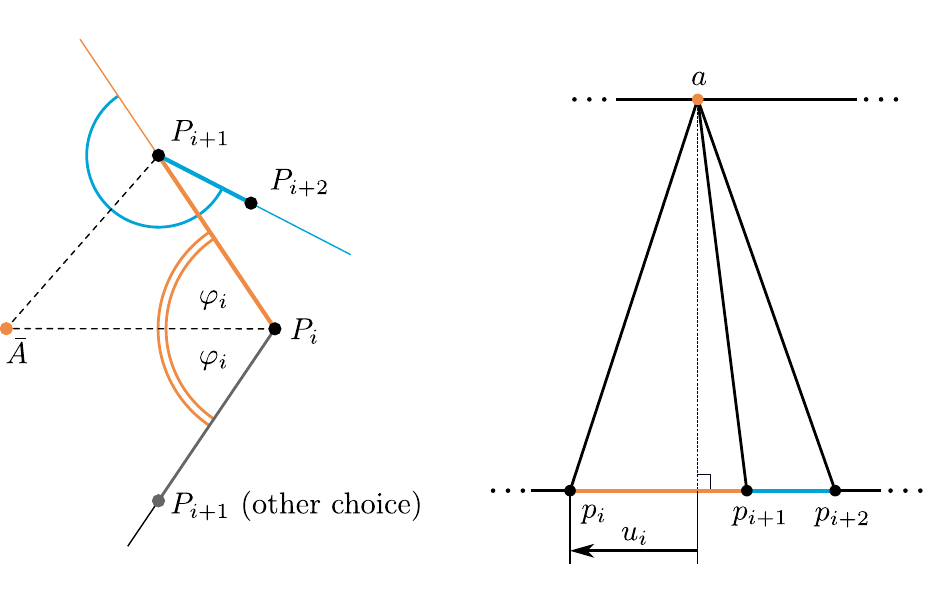}
	\caption{To obtain $P_{i+1}$, draw rays which enclose an angle of $\varphi_i$ with the line $\bar{A}P_i$.}
	\label{fig:fig3}
  \end{minipage}
\end{figure}

\begin{lemma}
  \label{lem:angle}
  The angle $\varphi_i = \angle \bar{A}P_iP_{i+1} \in (-\pi, \pi]$ satisfies
  \begin{align}
	\label{eq:angle}
	\cos(\varphi_i) \mathrel{=} \frac{-u_i}{\sqrt{h^2-H^2+u_i^2}}\ 
  \end{align}
 for all $i = 1,\dots, n-1$, where $u_i$ is the signed \emph{horizontal} distance in the material from $a$ to $p_i$, see Figure~\ref{fig:fig3}.
\end{lemma}
\begin{proof}
  The triangle $\bar{A}P_iP_{i+1}$ has side lengths $|P_iP_{i+1}|=u_{i+1}-u_i$, $|P_i\bar{A}|=\sqrt{h^2-H^2+u_i^2}$, and $|P_{i+1}\bar{A}| = \sqrt{h^2-H^2+u_{i+1}^2}$. Our claim is then equivalent to the law of cosines.
\end{proof}

As a consequence of this lemma the embedded point $P_{i+1}$ must lie on one of two rays that originate from $P_{i}$ and that enclose an angle of $\pm\varphi_i$ with the line  $P_{i}\bar{A}$, see Figure~\ref{fig:fig3}. Notice furthermore that folds with zero dihedral angle necessarily give rise to an isometric embedding. Therefore, one of the two rays originating from $P_{i+1}$ containing $P_{i+2}$ must be the continuation of the line segment $P_iP_{i+1}$, see Figure~\ref{fig:fig3}. Thus, the turning angle of the polygonal line $(P_1, \dots, P_n)$ at vertex $P_i$, $2\leq i \leq n-1$, is either zero or given by $\pm 2\varphi_i$. In summary, we obtain:

\begin{lemma}
Using the above notation, the turning angle of an embedded fan's boundary at an interior vertex $P_i$ is either $0$ or equal to $\pm2\varphi_i$, where $\varphi_i$ is defined by~\eqref{eq:angle}.
In particular, when prescribing a fan's embedded height $H$, there are only finitely many possible isometric embeddings of a fan up to rigid body motions.
  \label{lem:turningangle}
\end{lemma}

As an illustrative example consider the case $n=3$, where there exists only a single interior vertex, and for which there are exactly two possible embeddings up to rigid motions and reflections across planes, see Figure~\ref{fig:embeddingexample}.

We have thus determined the possible turning angles at the interior vertices of each embedded fan of prescribed height $H$. It remains to relate the turning angles at \emph{exterior} vertices, i.e., turning angles between consecutive downward (or upward) fans, to the embedded height $H$. 

\begin{figure}[t!]
	\centering
	\newcommand{\factor}{0.48}
	\begin{minipage}[t]{\factor\linewidth}
	  \includegraphics[width=\linewidth]{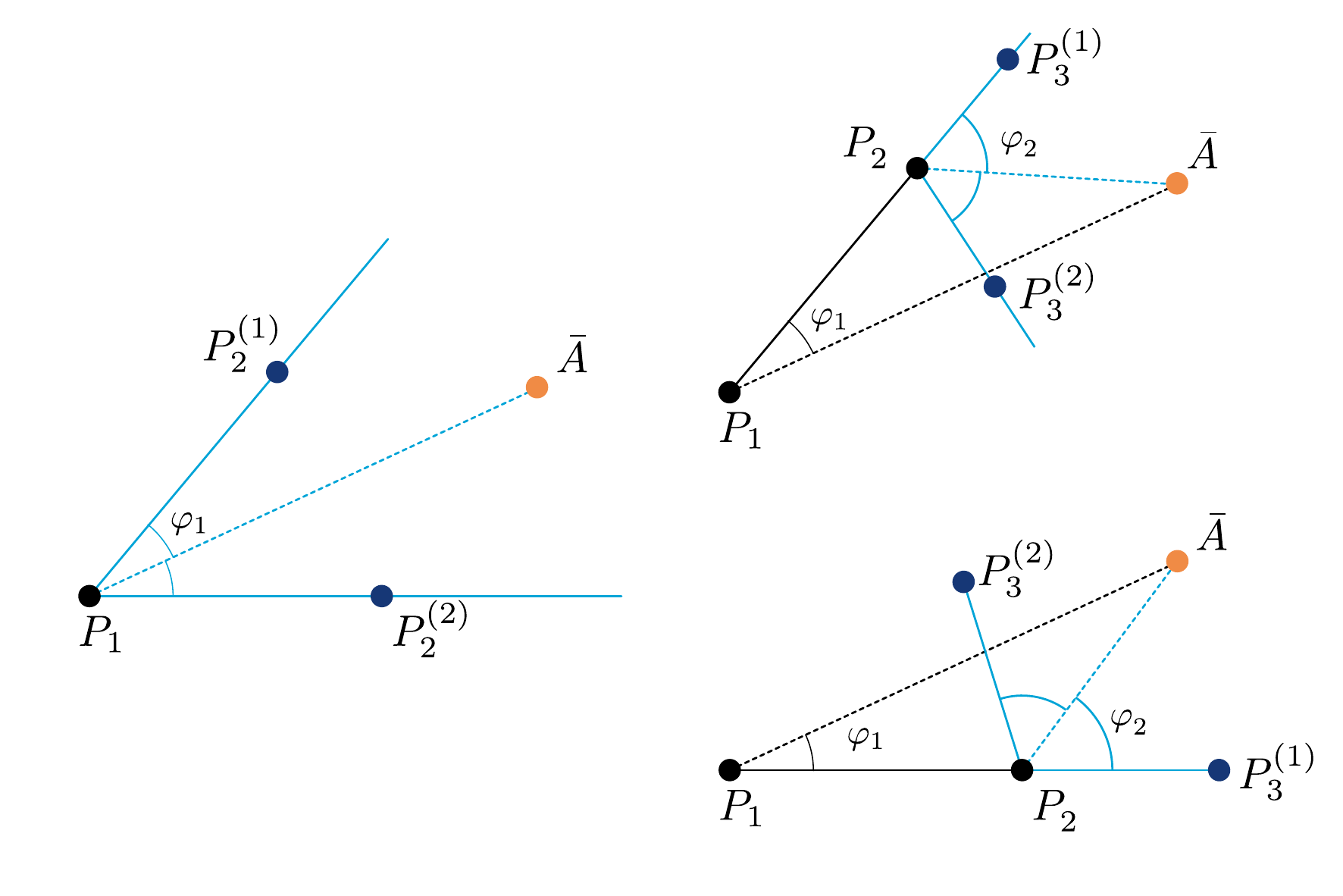}
	  \caption{For $n=3$, there are four different possibilities ($P_2$ and $P_3$ can be placed at two locations each). Since the construction is symmetric about the line $P_1\bar{A}$, there exist only two 'properly distinct' possible embeddings.}
	  \label{fig:embeddingexample}
	\end{minipage}
	\hfill
	\begin{minipage}[t]{\factor\linewidth}
	  \includegraphics[width=\linewidth]{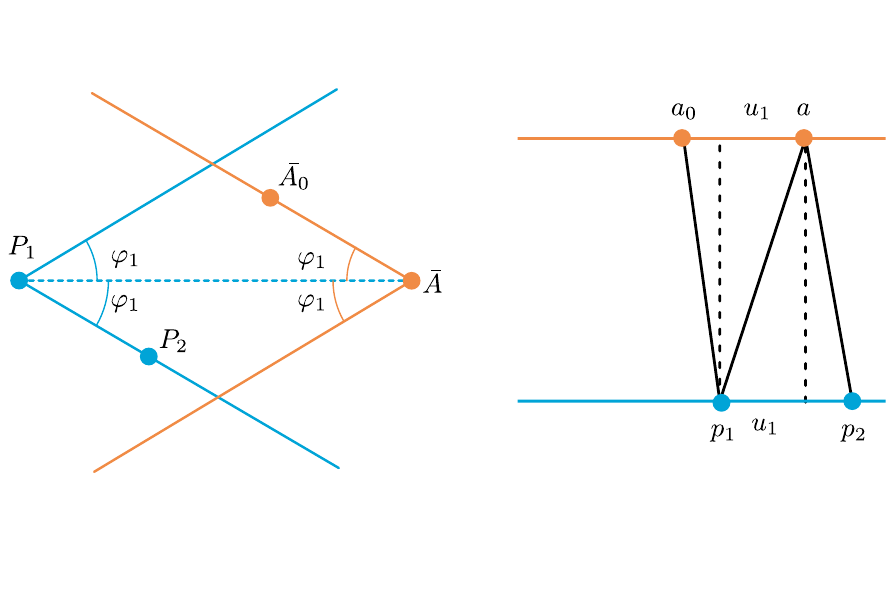}
	  \caption{The lines $P_1P_2$ and $\bar{A_0}\bar{A}$ enclose the same angle with $P_1\bar{A}$. Hence, they are either parallel or intersect at an angle of $\pm2\varphi_i$.}
	  \label{fig:parallelity}
	\end{minipage}
\end{figure}

\subsection{Turning angles at exterior vertices}
Fix a downward fan with apex $a$ and denote by $a_0$ and $a_2$ the vertices adjacent to $a$ on the strip's upper boundary, with their embedded locations denoted by $A,A_0$, and $A_2$, respectively. As above, denote by $p_1, \dots, p_n$ the vertices of the fan with apex $a$, with embedded positions $P_1, \dots, P_n$, and let $\varphi_i$ be the angles given by Lemma~\ref{lem:angle}. (Again, for simplicity we assume that $A$ lies on the $z=H$ plane and that $P_1$ through $P_n$ lie in the $(xy)$-plane.) Denote by $\bar{A_0}$, $\bar{A}$, and $\bar{A_2}$ the vertical projections of the points $A_0$, $A$, and $A_2$, respectively, onto the $(xy)$-plane. The following result relates the exterior turning angle $\alpha = \pi - \angle A_0 A A_2 = \pi - \angle \bar{A_0} \bar{A} \bar{A_2}$ of the strip's upper boundary to the angles $\varphi_i = \angle \bar{A}P_iP_{i+1}$ at the strip's lower boundary.

\begin{figure}[t]
  \centering
  \newcommand{\factor}{0.493}
	\includegraphics[width=\linewidth]{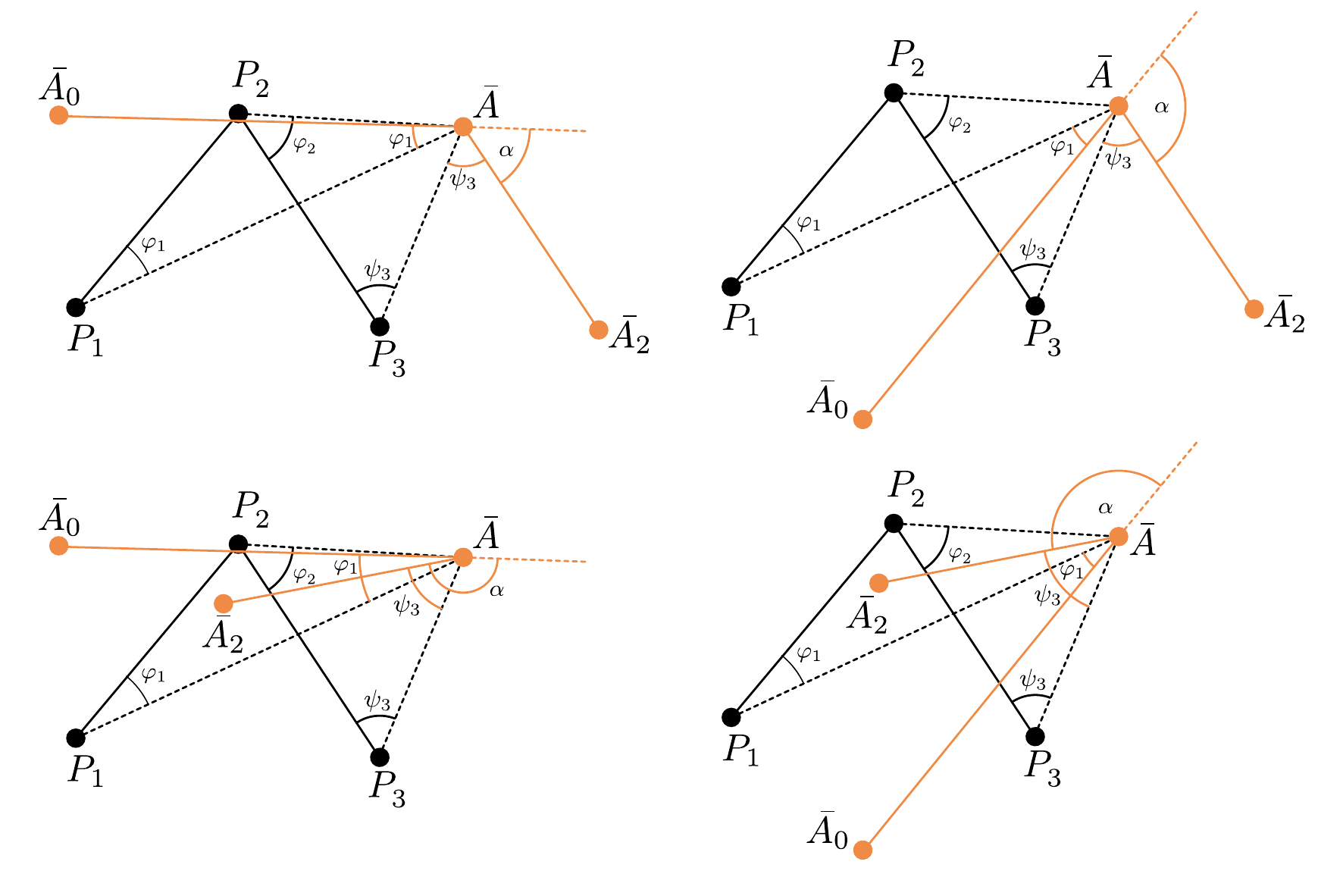}
	\caption{For one of the embeddings for $n=3$ depicted in Fig.~\ref{fig:embeddingexample}, there are four different ways to embed the points $\bar{A_0}$ and $\bar{A_2}$, resulting in four different possible values of $\alpha$.}
	\label{fig:alpha}
\end{figure}
\begin{lemma}
For a fixed downward fan with apex $a$, let $\beta$ denote the angle between the lines $P_1P_2$ and $P_{n-1} P_n$ and let $\psi_n:=\pi-\varphi_n$. Then the angle $\alpha= \pi - \angle A_0 A A_2$ satisfies $(\alpha + \beta) \in \{0, \pm 2\varphi_1, \pm 2 \psi_n, \pm 2 \varphi_1 \pm 2\psi_n\}\mod 2\pi$.
\end{lemma}
\begin{proof}
Observe that $a$ is an exterior vertex of the upward fan with apex $p_1$.
By applying Lemma~\ref{lem:angle} to the upward fan with apex $p_1$, we find that $\angle P_1\bar{A}\bar{A_0}= \pm \varphi_1$, see Figure~\ref{fig:parallelity}.
With reference to Figure~\ref{fig:parallelity}, it follows that the two lines $P_1P_2$ and $\bar{A}\bar{A_0}$ are either parallel or enclose an angle of $\pm2\varphi_1$. Likewise, the two lines $P_{n-1}P_n$ and $\bar{A}\bar{A_2}$ are either parallel or enclose an angle of  $\pm2\psi_n$. Together, one obtains four possible configurations for the polygonal line segments $A_0AA_2$, which (up to sign) lead to four possible values for the angle $\alpha$, see Figure~\ref{fig:alpha}. 
\end{proof}

Since the angle $\beta$ in the previous lemma between the lines $P_1P_2$ and $P_{n-1} P_n$ can be expressed (modulo $2\pi$) as a sum of the turning angles between the lines $P_{i-1}P_i$ and $P_{i}P_{i+1}$, $2\leq i \leq n-1$, and since these turning angles are either zero or given by $\pm2\varphi_i$, we obtain that $\alpha= \pi - \angle A_0 A A_2$ satisfies
\begin{align}\label{eq:alpha}
\alpha = \sum_{i=1}^{n-1} s(i) 2\varphi_i + s(n) 2 \psi_n \ , 
\end{align}
where $s(j)  \in \{-1,0,1\}$ and the exact value can be obtained from the preceding discussion.

In summary, using Lemma~\ref{lem:turningangle} and Equation \eqref{eq:alpha}, every possible value of every turning angle of both of a strip's boundaries can be expressed in terms of parameters known from the fold pattern and the strip's embedded height $H$ alone. Note in particular that all of these angles are analytic functions of $H$.

Our discussion shows that turning angles (interior and exterior) can only attain a finite number of values, which each depend analytically on $H$. Thus, since the set of values of turning angles at a given height $H$ is discrete, it follows that if an origami cylinder's height is changed continuously, then each of the resulting embeddings uniquely determines all of the turning angles for all other embeddings during the deformation.

\section{Cylinder compression}
\label{sec:can_compression}
Using the notation and results established above, we show that a large class of origami cylinders \emph{cannot} be compressed isometrically. In particular, our results prove that the fold patterns considered by Yasuda and Yang~\cite{Yasuda2015} are one of the very few patterns that allow for rigid foldability. For convenience, we repeat the statement of Theorem~\ref{thm:compressibility} from Section~\ref{sec:cans}:

\begin{figure}[t]
	\centering
	 \includegraphics[width=0.6\linewidth]{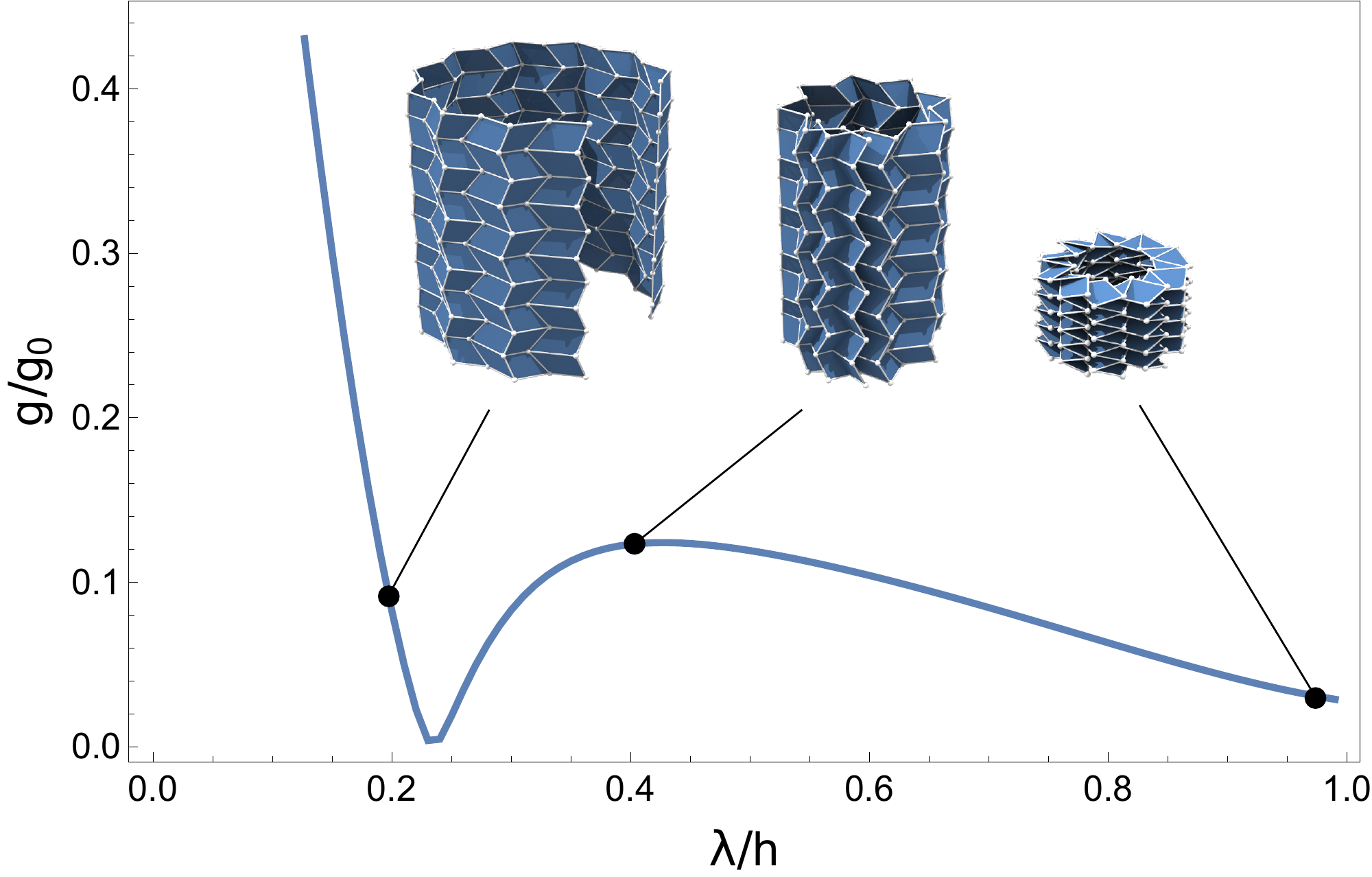}
	\caption{Plot of normalized gap magnitude as a cylindrical Miura pattern is compressed. We prove this function vanishes at only finitely many values of $\lambda$, and so the pattern is not rigid-foldable.}
	\label{fig:magnitudeplot}
\end{figure}

\begin{repthm}{thm:compressibility}
 An origami cylinder whose fold pattern does not contain vertical fold lines has at most a finite number of heights at which it can be isometrically embedded. Hence, the absence of vertical folds prohibits isometric deformations with continuously varying heights.
\end{repthm}
\begin{proof}
  We show that Theorem~\ref{thm:compressibility} holds for a single strip that contains no vertical fold line.
 Recall that for a given fold pattern and a given embedded height $H$, each of the finitely many \emph{open} strip embeddings can be encoded by the choices of the turning angles according to Lemma~\ref{lem:turningangle} and Equation~\eqref{eq:alpha}. Given these angles, the requirement of isometry shows that the position of the embedded strip's bottom-right corner $C_2$ is uniquely determined by the position of the bottom-left corner $C_1$. A necessary condition for the embedded strip to close up is $C_1 = C_2$; one measure of the failure of the strip to close up is the (squared) \emph{gap length} $g=|C_2-C_1|^2$, which, according to the previous section, is a function of the embedded height $H$ alone. As an example, Figure~\ref{fig:magnitudeplot} illustrates the function $g$ for a simple Miura-type origami that has been rolled up into a cylindrical shape.

The key observation is that for every embedding of the strip, the corresponding gap $g$ is a \emph{real analytic} function of the embedded height, $H$. Indeed, by summing up the individual line segment vectors of either of the strip's boundaries, $g$ can be expressed via the boundary's turning angles, which are analytic in $H$ according to the preceding section. Since the existence of a cylindrical embedding at some height $H$ necessitates that $g(H)=0$, it follows in particular that if the cylinder were continuously deformable from $H_0$ to $H_1$, then there would be a choice of turning angles for which $g$ would vanish on the entire interval $[H_0,H_1]$. Because $g$ is an analytic function, this is equivalent to $g(H)$ being zero for \emph{every} value of $H$.

Consider now the limiting case $H\to h$. By continuity, the resulting origami cylinder does not contain curved surface pieces. 
In this case, the assumption that the strip has no vertical fold line implies that the only possible embedding of the strip is the rigid one, i.e., the embedding of the material as a flat rectangle of size $l\times h$ perpendicular to the $(xy)$-plane. It follows that $g(h) = l \neq 0$. This shows that $g$ cannot be identically zero on some open interval; instead, its zeros are isolated. Hence, the set of all heights that allow for an isometric cylindrical embedding is finite.
\end{proof}

As an additional consequence of our discussions we conclude that vertical fold lines -- such as those appearing in the
construction by Yasuda and Yang~\cite{Yasuda2015} -- must have a dihedral angle that is either equal to 0 or equal to $\pm\pi$.
\begin{obs*}
  Each fold whose fold line in the fold pattern is vertical must have dihedral angle $0$ or $\pi$ whenever $H<h$. 

  Furthermore, in each compressible origami cylinder, there must exist at least two folds with vertical fold lines whose dihedral angles are constantly equal to $\pm\pi$.
  \label{thm:verticalfolds}
\end{obs*}
\begin{proof}
  Consider a triangulated strip fold pattern that contains a vertical fold line. This vertical fold line is a cathetus of both of its incident triangles. The other two catheti must be horizontal and thus embed to horizontal lines (by assumption (iii)). If the dihedral angle at the vertical fold line is not zero or $\pm\pi$, then these two catheti are linearly independent, implying that the fold corresponding to the vertical fold lines has to be vertical as well (since it must then be perpendicular to the entire horizontal plane). But then we must have that $H=h$. This shows that all folds with vertical fold lines must be either flat (zero dihedral angle) or sharp (dihedral angle $\pm\pi$).

 For the second claim, consider a compressible origami cylinder strip. The gap functions are analytic and thus they are constantly zero, hence they are also zero in the limit $H\to h$. In this limit, all folds with non-vertical fold lines have zero dihedral angle. But since the strip is still closed in this limit, there must exist folds with nonzero dihedral angle---which are precisely the vertical folds with dihedral angle $\pm\pi$.
\end{proof}

\section{Summary}
We have investigated the compressibility of cylindrical origami structures and have shown that in general such structures cannot be compressed isometrically. The fact that structures such as those presented in Figure~\ref{fig:teaser} and Figure~\ref{fig:nonbellows} nevertheless appear to be isometrically compressible suggests that the requisite stretching is either so diffused as to be invisible at the macroscopic scale or localized to specific regions---to the effect that macroscopic deformations of the material remain small. Simulations of elastic ridges done by Witten and Lobkovsky~\cite{Lobkovsky1996,Witten2007a} support the latter explanation. The principles leading to such localized stretching and their applicability to the design of foldable origami will be the subject of future work.

\section{Acknowledgements}
F.B. and M.W. acknowledge the support of the BMBF collaborative project MUSIKA. E.V. acknowledges the support by NSF grant DMS-1304211. We thank Shmuel Rubinstein for insightful discussions about our paper folding experiments.
\bibliographystyle{vancouver}
\bibliography{Origami}





\newpage
\appendix

\section{Regularity of the embedding}
\begin{figure}[t]
  \centering
  \newcommand{\factor}{0.493}
  \includegraphics[width=\factor\linewidth]{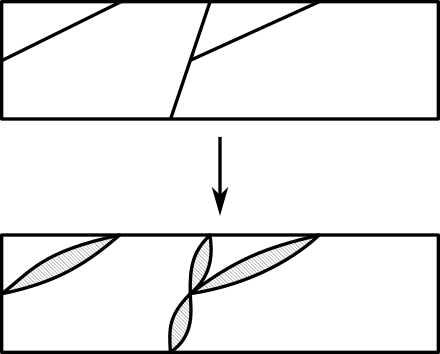}
  \caption{The initial rectangular domain (top) is split along its fold lines to distinguish the two sides of each fold (bottom). Hatched regions no longer belong to the domain.}
  \label{fig:regularity}
\end{figure}
In order to make precise our definition of \emph{piecewise $C^4$}, we
first remove the (finite number of) fold lines from the original material domain. Now consider the topological closure of the resulting set such that fold lines are doubled---with the two copies of each original fold line only identified at their endpoints, i.e., the new domain is a planar rectangle whose boundary consists of the original material boundary and two distinct line segments for each original fold line, running between the original fold line's end-points. See Figure~\ref{fig:regularity} for an example.
We require the embedding map $f$ to be a $C^4$-isometry (up to and including the boundary) on this new domain. By our requirements, fold lines get mapped to straight line segments; hence, the two copies of each original fold line must occupy the same position in space, such that the resulting surface does not suffer from any additional holes. Notice that a $C^4$-isometry (up to and including the boundary) excludes cone points, i.e., isolated singular points on the cylinder surface, and that it also implies that the cylinder is composed of flat polygons:

It is surprising that this result would require such a high regularity of the embedding when all notions required for developability such as e.g. curvature are already perfectly well-defined for surfaces of $C^2$ regularity. However, to the best of our knowledge, there exists no proof of Theorem~\ref{thm:flatness} that does not require the aforementioned $C^4$ regularity.

Using the above definitions, we now formalize our main result as: for any \emph{fixed} fold pattern, does there exist a corresponding one-parameter family of embeddings to origami cylinders of continuously varying axial height?

\section{Flatness of facets}
This section complements the main result by showing that, even though it is only required that fold lines are mapped to straight line segments, the intermediate pieces of the cylinder's surface have to be flat as well, thereby proving Theorem~\ref{thm:flatness}:

\begin{repthm}{thm:flatness}
  Every origami cylinder whose embedded height is strictly less than the material height consists entirely of planar faces that meet at the embedded fold lines. In particular, there are no creases in the embedding apart from those given by the fold lines.
\end{repthm}

Let \emph{facets} denote those connected components of the origami cylinder that arise when all folds are removed.
Internal facets are those that do not meet either of the cylinder's boundaries except at isolated vertices.
\subsection{Internal facets}
 Our requirements imply that internal facets must be flat:
\begin{lemma}
  \label{lem:internalflatness}
	Let $P$ be a closed non-intersecting polygon with finitely many corners in the plane, which is embedded isometrically via a $C^4$ map $f$ into $\mathbb{R}^3$ such that each boundary edge of $P$ is mapped to a straight line segment in $\mathbb{R}^3$. Then the resulting surface $S$ is contained in a plane.
\end{lemma}
\begin{proof}
  Since $S$ is developable, its Gauß curvature vanishes everywhere. Let parabolic points be those points of $S$ with nonvanishing mean curvature and call all other points flat. Asymptotic curves are those curves in $S$ whose tangent vectors coincide with directions of vanishing directional curvature. Notice that through each parabolic point there exists a unique (up to reparameterization and extension) asymptotic curve.
  We require the following two facts about developable $C^4$ surfaces~\cite{Massey1962}:
  \begin{enumerate}
	 \item[(i)] Every asymptotic curve that originates in a parabolic point cannot contain flat points.
	\item[(ii)] If an asymptotic curve contains a parabolic point, then it is a straight line segment.
  \end{enumerate}
Suppose that $S$ contained a parabolic point. By continuity of curvature, the set of parabolic points is open in $S$, so there exists a curve $\Gamma$ of nonzero arc length in $S$ that consists entirely of parabolic points and that intersects asymptotic curves transversally. Then (i) and (ii) imply that $S$ contains a family of straight line segments (consisting entirely of parabolic points), originating from each point of $\Gamma$. Each of these line segments intersect the boundary of $S$ in (at least) two different points which vary continuously with their origin on $\Gamma$. Since not both of these intersections can be constant along $\Gamma$ (if they were, then all asymptotic curves would be the same, thereby contradicting the fact that $\Gamma$ is transversal and has nonzero length), there exists a point on $\Gamma$ whose asymptotic line meets the boundary of $S$ transversally at a point that is not a corner. But then this point of intersection contradicts (i) since it is flat. Hence $S$ consists entirely of flat points.
 \end{proof}
In order to treat the remaining facets, i.e., those that meet the origami cylinder's boundary, we require a property about the reconstruction of developable ruled surfaces from maximal geodesics that are transversal to rulings.

\subsection{Reconstruction of a ruled surface from a known geodesic}
We call a ruled surface $S$ \emph{generated} by a
 \emph{directrix curve} $\gamma$ (with $ |\dot \gamma(t)| =1$ for all $t$) and a \emph{director curve} $v$ (with $|v(t)|=1$ for all $t$) if $S$ is given by $\{\gamma(t)+ sv(t) \}$.
\begin{thm}
	Let $S, \tilde S$ be $C^3$ developable ruled  surfaces that are generated by a directrix curve $\gamma$ and two director curves $v, \tilde v$. If $\ddot\gamma(t)\neq 0$ for all $t$ and $\gamma$ is a geodesic in both $S$ and $\tilde S$ then $v$ and $\tilde v$ must agree up to sign.
\label{thm:developable}
\end{thm}
\begin{proof}
The surface  $S$ can be parameterized via $f(s,t) = \gamma(t) + s v(t)$.
From the fact that $\gamma$ is a geodesic, it follows that $\ddot\gamma$ is normal to $S$. Let $N \mathrel{:=} \ddot\gamma / |\ddot\gamma|$. Notice that $N(t)$ is normal to $S$ along the ruling $\{\gamma(t) + s v(t)\}$. The ruling direction is obviously a tangent direction of $S$, so that
	\begin{align}
		v &= \cos\varphi \cdot\dot\gamma + \sin\varphi \cdot(N\times\dot\gamma)
		\label{eq:v}
	\end{align}
	with an angle function $\varphi$.

The second fundamental form of $S$ is given by
	\begin{align*}
		\II &=
		\begin{pmatrix}
			\langle f_{tt},N\rangle & \langle f_{st},N\rangle \\
			\langle f_{ts},N\rangle & \langle f_{ss},N\rangle
		\end{pmatrix}\ .
	\end{align*}
	Developability implies that the determinant of $\II$ vanishes on an open and dense subset of $S$. Since $f_{ss}$ is zero, we conclude that the off-diagonal entries of $\II$ must also be zero. Hence, the vector $f_{st} = \dot v$ is perpendicular to the surface normal.

Taking derivatives of Equation~\eqref{eq:v} with respect to $t$ gives
	\begin{align*}
		\dot v = -&\dot\varphi\sin\varphi\cdot\dot\gamma + \cos\varphi\cdot\ddot\gamma + \\
		&\dot\varphi\cos\varphi\cdot(N\times\dot\gamma) + \\
		& \sin\varphi\cdot(\dot N\times\dot\gamma+N\times\ddot\gamma)\ .
	\end{align*}
	We have $N\times\ddot\gamma =0$ since $N$ is parallel to $\ddot\gamma$. Furthermore, the component of $\dot v$ along the surface normal must vanish:
	\begin{align*}
	   0 &= \cos\varphi\langle\ddot\gamma,N\rangle + \sin\varphi\langle\dot N\times\dot\gamma,N\rangle\ .
	\end{align*}
   To deal with the first term, notice that
  $ \langle\dot\gamma, N\rangle = 0$, which implies that we have $\langle\ddot\gamma,N\rangle + \langle\dot\gamma,\dot N\rangle = 0$. This turns the above equation into
   \begin{align*}
	   0 &= \cos\varphi\langle\dot\gamma,\dot N\rangle + \sin\varphi\langle \dot\gamma\times\dot N, N\rangle\ .
   \end{align*}
   Since $\ddot \gamma(t) \neq 0$, the angle $\psi \mathrel{:=} \angle(\dot N,\dot\gamma)$ can be obtained from $\gamma$ alone. Hence, the above equation turns into
   \begin{align*}
	   0 &= \cos\varphi\cos\psi + \sin\varphi\sin\psi \\
	   \Leftrightarrow 0 &= \cos(\varphi-\psi)\ .
   \end{align*}

   From the knowledge of $\psi$ we can infer $\varphi$, which, together with $N$, uniquely determines $v$.
\end{proof}

\subsection{Boundary facets}
Even though our origami cylinder only contains line segments as fold lines, it is not a priori clear that the cylinder's boundaries are also polygonal lines. In this section, we show that they must nevertheless be piecewise straight if the cylinder's embedded height is strictly smaller than the material height $h$. In any case, we have already shown that facets which are completely enclosed by fold lines (i.e., completely enclosed by straight line segments) must be flat. We now extend this result to those facets that meet the cylinder's boundary and prove Theorem~\ref{thm:flatness}.
\begin{proof}
  For a given origami cylinder whose embedded height is (strictly) less than the material height, consider a segment of either boundary that is at least $C^4$ (i.e., does not contain vertices). Let $\gamma$ and $\Gamma$ be the corresponding material and embedded curves of such a boundary segment, and assume, by contradiction, that $\Gamma$ is curved (i.e., its second derivative is nowhere zero---this can always be achieved by restriction to a part of the segment). Since $\Gamma$ consists entirely of parabolic points, each point carries a unique asymptotic line, whose collection is a ruled developable surface to which we apply the previous theorem.

By our definition of origami cylinders, $\Gamma$ lies in a horizontal plane and thus satisfies $\psi\equiv 0$, implying $\varphi\equiv\tfrac{\pi}{2}$, where $\psi$ and $\varphi$ refer to the angle functions in the proof of Theorem~\ref{thm:developable}. This and the fact that the curvature vector of $\Gamma$ is parallel to the $(xy)$-plane imply that the asymptotic lines emanating from $\Gamma$ are strictly vertical.
By property (i) above, these asymptotic lines cannot contain flat points and thus cannot meet fold lines transversally. Since the number of fold lines is finite, this implies that there exists an asymptotic line that traverses the cylinder from top to bottom. But then the cylinder's height equals the material height, contradicting the initial assumption.
\end{proof}

As a corollary of Theorem~\ref{thm:flatness}, we can conclude that a fold pattern must be embedded in such a way that there are no extra creases beyond those that originate from the fold lines:
\begin{corollary}
  \label{cor:no-extra-creases}
  Every fold pattern can only be embedded in such a way that its image consists entirely of planar pieces that must meet at the images of the fold lines.
\end{corollary}

\section{Construction of fans from given exterior vertices}
\subsection{A fold pattern where no short embedding can exist}
The next section will show how to complete short embeddings of fans (i.e., those where the distance between the exterior vertices' embedded positions is less than their material distance) into isometric ones by introducing additional fold lines to the fan.

However, it is not even clear that every fold pattern has a short embedding. Consider a single strip with three vertices at unit distance apart on each boundary, skewed into a truss-like shape. Of this fold pattern, we show that it is impossible to find an embedding whose height is half the material height which leaves the fold lines at the right length and compresses all boundary segments.

We first examine one of its downward fans (which are here actually only a single triangle) and its adjacent triangles. If the turning angle at the top vertex, $\alpha$, and the two boundary segments incident to the apex vertex are compressed to lengths $l_1,l_2\leq 1$, then one can find two possible embeddings each for the two bottom vertices. Let $e$ be the squared Euclidean distance between these embedded points. For each of the four choice pairs, $e$ is a function of $\alpha,l_1$, and $l_2$, where $0<\alpha<\pi$ and $0<l_1,l_2\leq1$. Then one can show that for $\tfrac{\pi}{5}\leq\alpha<\pi$, all four versions of $e$ are strictly positive for all values of $l_1,l_2$. Conversely, if the fan itself is embeddable, then $\alpha<\tfrac{\pi}{5}$.

If there existed a nondegenerate embedding, the top boundary would necessarily embed into a triangle. However, we have just shown that in this triangle, each turning angle is bounded by $\tfrac{\pi}{5}$, which is clearly impossible (the sum of all three turning angles needs to be equal to $2\pi$). Thus we have shown that the fold pattern described cannot be embedded in such a way that the embedding can be completed to an isometric embedding by the introduction of new folds.

However, there is a large class of strip patterns for which this is possible: those where one can pick three vertices each from top and bottom boundary such that the six vertices alternate between top and bottom boundary when sorted according to their horizontal position.

\subsection{Completing short embeddings to isometric embeddings}
  Contrary to the use case of finding an embedding for a given fold pattern, it is also useful to construct embedding and fold pattern at the same time in order to obtain special embeddings such as rotationally symmetric cylinders. This section will therefore provide a means of extending a basic ``skeleton'' fold pattern with a corresponding embedding to one that is additionally isometric by successive addition of fold lines. It turns out that this extension can already take place on the level of individual fans, with very little conditions on the skeleton pattern.

  Suppose that the apex vertices $a$ and $A$ and the exterior vertices  $p_1$, $p_n$ and $P_1$, $P_n$ of a fan are given in the material and spatial domain, respectively. Below we offer a construction for the interior vertices $p_2, p_3, \dots, p_{n-1}$ as well as their embedded locations $P_2, P_3, \dots, P_{n-1}$
  such that the resulting embedding is isometric. For this construction to work, the \emph{only} condition to be met is that $|P_nP_1| \leq |p_np_1|$ as shown by the following result.

  \begin{thm}
	\label{thm:refinement}
	Consider a triangle $\Delta(abc)$ in the plane and a triangle $\Delta(ABC)$ in $\mathbb{R}^3$. Let $E$ be any plane containing $B$ and $C$. If the following conditions are satisfied
	\begin{itemize}
	  \item $|AB| = |ab|$,
	  \item $|AC| = |ac|$,
	  \item $|BC| \leq |bc|$,
	  \item $H \leq h$, where $H$ is the distance from $A$ to the plane $E$ and $h$ is the distance from $a$ to $bc$,
	\end{itemize}
	then there exists a point $D$ on $E$ such that the triangle pair $\Delta(ABD) \cup \Delta(ADC)$ is isometric to $\Delta(abc)$.
  \end{thm}

\begin{figure}[t]
  \centering
  \newcommand{\factor}{0.493}
  \begin{minipage}[t]{\factor\linewidth}
	\centering
	\includegraphics[width=\linewidth]{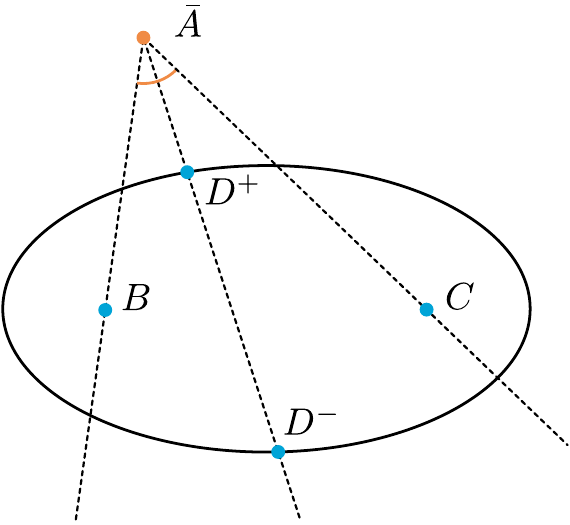}
	\caption{On the ellipse with foci $B$ and $C$, the points $D^+$ and $D^-$ lie on the intersection of the ellipse with the angular bisector of the angle $\angle B\bar{A}C$.}
	\label{fig:ellipse}
  \end{minipage}
  \hfill
  \begin{minipage}[t]{\factor\linewidth}
	\centering
	\includegraphics[width=\linewidth]{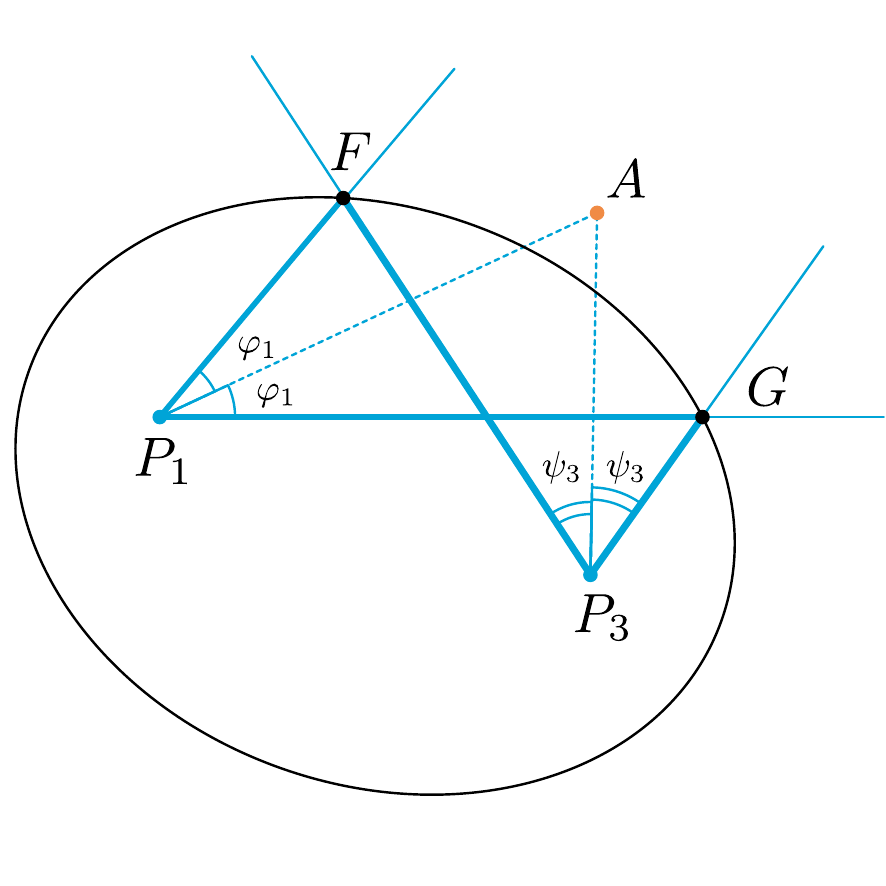}
	\caption{For $n=3$, the single interior vertex $p_2$ can be embedded at only two different locations $F$ and $G$, which are marked by the intersection of pairs of rays.}
	\label{fig:fanconstruction}
  \end{minipage}
\end{figure}

  \begin{proof}
	Consider the set of all points $D\in E$ with $|BD|+|DC| = |bc|$. This locus of points forms an ellipse $\eta$ in the plane $E$ with foci $B$ and $C$. It suffices to show that for some $D$ on this ellipse, $\angle BDA + \angle ADC = \pi$, for if this is the case, then the hinge composed of the two triangles $\Delta(ABD)$ and $\Delta(ADC)$ can be isometrically unfolded into a triangle congruent to $\Delta(abc)$.

	Let $\bar A$ be the orthogonal projection of $A$ onto $E$. First, $\bar A$ lies on or outside of the ellipse $\eta$:
	\begin{align*}
	  |\bar AB| &+ |\bar AC| = \\
	  &= \sqrt{|AB|^2-H^2}+\sqrt{|AC|^2-H^2}\\
	  &= \sqrt{|ab|^2-H^2}+\sqrt{|ac|^2-H^2}\\
	  &\geq \sqrt{|ab|^2-h^2} + \sqrt{|ac|^2-h^2}\\
	  &\geq |bc|,
	\end{align*}
	precisely the condition that $\bar A$ does not lie inside $\eta$.

	Now consider the angle bisector of the angle $\angle B\bar AC$, and denote its two intersection points with $\eta$ by $D^+$ and $D^-$, where $D^+$ is closer to $\bar A$, see Figure~\ref{fig:ellipse}. Then
	\begin{align*}
	  \angle BD^+A + \angle AD^+C \geq \pi\ ,\\
	  \angle BD^-A + \angle AD^-C \leq \pi\ .
	\end{align*}
	In order to prove, e.g., the first line, express the cosine of the two angles via the inner product of the vectors $(B-D^+)$, $(C-D^+)$, and $(A-D^+)$:
	\begin{align*}
	  \cos\angle BD^+A &=\frac{\langle B-D^+,A-D^+\rangle}{|B-D^+|}\frac{1}{|A-D^+|}\ ,\\
	  \cos\angle AD^+C &=\frac{\langle C-D^+,A-D^+\rangle}{|C-D^+|}\frac{1}{|A-D^+|}\ .
	\end{align*}
	Since $(B-D^+)$ and $(C-D^+)$ are horizontal, the only quantity that depends on $H$ is the length $|A-D^+|$. We leverage this fact to conclude that the cosines' sign is independent of $H$. For the sake of brevity, write
	\begin{align*}
	  \cos\angle BD^+A &=\frac{c_1}{\sqrt{c_3^2+H^2}}\ ,\\
	  \cos\angle AD^+C &=\frac{c_2}{\sqrt{c_3^2+H^2}}\ .
	\end{align*}

	Proving that $ \angle BD^+A + \angle AD^+C \geq \pi$ is equivalent to proving that
	\begin{align*}
	  \sin\left(  \angle BD^+A + \angle AD^+C \right)\leq 0
	\end{align*}
	since $\angle BD^+A + \angle AD^+C <2\pi$ by construction. Since each individual angle is bounded above by $\pi$, we use $\sin(x)=\sqrt{1-\cos(x)^2}$ to obtain
	\begin{align}
	  &\sin\left(  \angle BD^+A + \angle AD^+C \right) = \nonumber\\
	  &\frac{\left(  c_2\sqrt{c_3^2-c_1^2+H^2}+ c_1\sqrt{c_3^2-c_2^2+H^2} \right)}{c_3^2+H^2}\ .
	  \label{eq:sine}
	\end{align}
	To show that this expression is nonpositive, we show that it either vanishes for all $H$ or that it is negative for $H= 0$ and cannot vanish for any other $H$ in this case. Continuity then implies the desired conclusion.

	It is straightforward to verify (using the Cauchy-Schwarz inequality) that $c_3^2\geq c_1^2$ and  $c_3^2\geq c_2^2$. Hence, if the expression in~\eqref{eq:sine} is zero for some $H_0$, then either both $c_1$ and $c_2$ are zero or both are nonzero with different sign. If both are zero, then both angles $\angle BD^+A$ and $\angle AD^+C$ are equal to $\pi/2$ regardless of $H$ and the expression in~\eqref{eq:sine} vanishes for all $H$. If both are nonzero, rearranging terms in~\eqref{eq:sine} yields $c_1 + c_2 =0$, which is equivalent to $\angle BD^+A = \pi - \angle AD^+C$. This holds independently of $H$, so again the expression in~\eqref{eq:sine} vanishes for all $H$.

	Hence, the expression in~\eqref{eq:sine} either vanishes for all $H$ or does not vanish for any $H$. In the former case, there is nothing to show. Consider the latter case and let $H=0$. Then $A$ and $\bar A$ agree and $D^+$ lies in the inside or on the boundary of the triangle $\Delta(ABC)$, which necessitates $ \angle BD^+A + \angle AD^+C \geq\pi$ with equality if and only if the expression in~\eqref{eq:sine} vanishes, which we have ruled out by assumption.

	In order to finish the proof of Theorem~\ref{thm:refinement}, notice that $\angle BDA + \angle ADC$ is continuous as $D$ sweeps along an arc of the ellipse between $D^-$ and $D^+$; hence, there must exist a $D$ with $\angle BDA + \angle ADC = \pi.$
  \end{proof}
  \begin{rem}
	\label{rem:backwardsrays}
	There are at least two distinct such points $D \in E$ unless the ellipse is degenerate or $\bar A$ lies on $\eta$, which occurs when one of the inequalities in the theorem statement becomes tight: $|BC| = |bc|$ or $H=h$. In Figure~\ref{fig:fanconstruction}, the two solutions are the points $F$ and $G$ (see below for an explanation).
  \end{rem}

  The previous theorem also shows that as soon as $A$, $P_1$, and $P_n$ have been placed in such a way that $|P_1P_n|\leq |p_1p_n|$, it is always sufficient to introduce a single interior fold to obtain a valid isometric embedding of the material triangle $AP_1P_n$. This, together with Lemma~\ref{lem:angle}, provides a method of constructing fans with an arbitrary (but still finite) number of interior folds.

  Recall that from given $P_1$, the possible locations of $P_2$ are located on two rays originating from $P_1$. Where previously $P_2$ could in principle lie anywhere on that ray, now not only is $p_2$ no longer given, but the final point $P_n$ -- which was previously unconstrained -- is now prescribed. In other words, we need to consider the ellipse $\eta$ given by the set $\{P \ : \ |P_1P|+|PP_n| = |p_1p_n|\}$ and place $P_2$ at any point on the chosen ray inside $\eta$, also determining $p_2$ in the process. As an illustration, $P_2$ can be placed at any point that lies on one of the bold line segments $P_1F$ and $P_1G$ in Figure~\ref{fig:fanconstruction}. Notice that by assumption $|P_1P_n|~\leq~|p_1 p_n|$, which implies that
  \begin{align*}
  |P_2P_n| &\leq |p_1 p_n| - |P_1 P_2|\\
  &= |p_1 p_n| - |p_1 p_2|\\
  &= |p_2p_n| \ .
  \end{align*}

  Consequently, the conditions of Theorem~\ref{thm:refinement} are satisfied for all $i$ with $A=A$, $B=P_i$, $C=P_n$, $a=a$, $b=p_i$, and $c=p_n$ and the construction that led to $P_2$ can now be repeated for $P_3$ and so forth.

The statement of Theorem~\ref{thm:refinement} then has two consequences: (a) If we place $P_{i}$ \emph{on} the attendant ellipse during our construction, then $P_{i}$ corresponds to one of the points $D$ from the theorem and all subsequent vertices $P_{i+1}, \dots, P_n$ lie on a straight line, and (b) if, on the other hand, we place $P_{2}, P_3, \dots, P_{n-2}$ strictly \emph{inside} the attendant ellipses, then the existence of the point $D$ of the theorem guarantees that we can always construct the penultimate vertex $P_{n-1}$ to obtain an isometric embedding.

\begin{rem}
  Similarly, one can consider a more general (and more symmetric) variant of the above construction, by first placing the vertices $P_2, P_3, \dots, P_i$, then $P_{n-1}, P_{n-2}, \dots, P_{n-k}$, then placing $P_{i+1}, P_{i+2}, \dots, P_{i+j}$, and so forth, i.e., by alternating the placement of new vertices between left and right, until all vertices have been placed.
\end{rem}
\end{document}